\documentclass[conference]{IEEEtran}
\IEEEoverridecommandlockouts

\usepackage{cite}
\usepackage{threeparttable}
\usepackage{graphicx}
\usepackage{picinpar}
\usepackage[cmex10]{amsmath}
\usepackage{amsmath,amsfonts,amssymb}
\usepackage[ruled,lined,boxed]{algorithm2e}
\usepackage{subfigure}
\usepackage{threeparttable}
\usepackage{caption}
\usepackage{booktabs}
\usepackage{multirow}
\usepackage{etoolbox}
\usepackage{makecell}
\usepackage{amssymb}
\usepackage{mathrsfs}
\usepackage{color}
\definecolor{purple}{RGB}{160,32,240}
\definecolor{darkred}{RGB}{255,0,255}
\usepackage{stfloats}
\usepackage{bm}
\newtheorem{lemma}{Lemma}

\newtheorem{proof}{Proof}

\begin{document}

\title{
	Joint Active User Detection and Channel Estimation for Grant-Free NOMA-OTFS in LEO Constellation Internet-of-Things
	
}

\author{
\IEEEauthorblockN{Xingyu Zhou\IEEEauthorrefmark{1}, Zhen Gao\IEEEauthorrefmark{2}}
\IEEEauthorblockA{\IEEEauthorrefmark{1}School of Information and Electronics, Beijing Institute of Technology, Beijing 100081, China}
\IEEEauthorblockA{\IEEEauthorrefmark{2}Advanced Research Institute of Multidisciplinary Science, Beijing Institute of Technology, Beijing 100081, China}
\IEEEauthorblockA{E-mail: gaozhen010375@foxmail.com}
}

\maketitle

\begin{abstract}
	
The flourishing low-Earth orbit (LEO) constellation communication network provides a promising solution for seamless coverage services to Internet-of-Things (IoT) terminals. However, confronted with massive connectivity and rapid variation of terrestrial-satellite link (TSL), the traditional grant-free random-access schemes always fail to match this scenario.  	
In this paper, a new non-orthogonal multiple-access (NOMA) transmission
protocol that incorporates orthogonal time frequency
space (OTFS) modulation is proposed to solve these problems.
Furthermore, we propose a two-stages joint active user detection and channel estimation scheme based on the training sequences aided OTFS data frame structure. 
Specifically, in the first stage, with the aid of training sequences, we perform active user detection and coarse channel estimation by recovering the sparse sampled channel vectors. 
And then, we develop a parametric approach to facilitate more accurate result of channel estimation with the previously recovered sampled channel vectors according to the inherent characteristics of TSL channel.     
Simulation results demonstrate the superiority of the proposed
method in this kind of high-mobility scenario in the end.

\end{abstract}


\IEEEpeerreviewmaketitle

\section{Introduction}

With the spring up of the low earth orbit (LEO) satellite constellation communication network, it has been expected to complement and extend terrestrial communication networks for the sake of its wide scope of coverage.  
As a consequence, LEO satellite constellation also provides a promising solution for supporting seamless coverage services in the Internet of Things (IoT) \cite{IoT2,liu.LEO}.

In fact, confronted with the massive machine-type communication, the grant-free random-access (GF-RA) schemes are always favored in IoT for their low-complexity, low-latency, and high-reliability \cite{LEO.IoT}, where active IoT terminals share the non-orthogonal resource allocation and directly transmit their data packets and pilot sequences without applying for the grant.   
Then along comes one crucial task to detect IoT terminals' activity. 
By exploiting the intrinsic sporadic traffic property \cite{LEO.IoT}, several compressive sensing (CS)-based GF-RA schemes have been proposed, where active user detection (AUD) was formulated as a sparse signal recovery problem \cite{AUD}. 
In \cite{wang1,wang2}, two low-complexity multi-user detectors based on structured CS were proposed to jointly detect users' activity and transmit signal in several continuous time slots. 
Moreover, to further improve performance of the system, some joint AUD and channel estimation (CE) schemes on the basis of approximate message passing (AMP) were proposed in massive multiple-input multiple-output (MIMO) systems \cite{AMP1,AMP2}. 
Besides, for narrowband terrestrial–satellite GF-RA system, a Bernoulli–Rician message passing with expectation–maximization
(BR-MP-EM) algorithm was proposed to address the AUD and CE problem \cite{LEO.IoT}.
For the challenging massive MIMO orthogonal frequency division multiplexing (OFDM) systems, the authors in \cite{Ke.Massive access} developed a CS-based adaptive AUD and CE method by leveraging the sporadic traffic and the virtual angular domain sparsity of massive MIMO channels.

Most prior works mentioned above relied on the premise of slow-changing property of channel state, in other words it could be viewed as relatively static over the time intervals of interest. 
However, the high mobility of satellite inevitably leads to rapid variation of terrestrial–satellite link (TSL). 
As a result, the existing GR-RA schemes \cite{LEO.IoT,AUD,wang1,wang2,AMP1,AMP2,Ke.Massive access} fail to match the LEO constellation-enabled IoT. 

A new two-dimensional modulation scheme, orthogonal time frequency space (OTFS), is expected to be a promising alternative to the dominant OFDM to support reliable communication under high-mobility scenarios \cite{OTFS}.  
By multiplexing information symbols on a lattice in the Delay-Doppler (DD) domain and utilizing a compact DD channel model that explicitly represents a stable and deterministic geometry of the channel, OTFS can achieve efficient and accurate channel estimation and signal detection with additional diversity gain.
Therefore, motivated by these compelling characteristics of OTFS, 
we intend to introduce a novel new non-orthogonal multiple-access (NOMA) transmission protocol that incorporates OTFS modulation to LEO constellation-enabled IoT, where the information symbols transmitted by IoT terminals are arranged in DD plane and modulated by OTFS to combat the dynamics of the TSL.

Furthermore, we propose a joint AUD and CE method to detect activities and estimate channel for all potential terminals which share the same access resources based on the designed training sequences aided OTFS (TS-OTFS) data frame structure. 
The proposed method includes two steps.  
First of all, by exploiting sporadic traffic of IoT and the inherent sparsity of TSL channel, the AUD and coarse CE can be formulated as a sparse signal recovery problem, where the activity state of all terminals and the sparse sampled channel vectors can be recovered by solving it. 
Moreover, a parametric approach is developed to facilitate more accurate result of channel estimation with the previously recovered sampled channel vectors in accordance with the inherent characteristics of TSL channel.   

\textit{Notation}: Throughout this paper scalar variables are denoted
by normal-face letters, while boldface lower and upper-case letters
denote column vectors and matrices, respectively. $\left[ \mathbf{X} \right]_{m,n}$ is the $(m,n)$-th element of matrix $\mathbf{X}$; $\left[ \mathbf{X} \right]_{m,:}$ and $\left[ \mathbf{X} \right]_{:,n}$ are the the $m$-th row vector and the $n$-th column vector of matrix $\mathbf{X}$, respectively. The transpose, conjugate transpose and pseudo-inverse operators are denoted by $(\cdot)^{\rm T}$, $(\cdot)^{\rm H}$ and $(\cdot)^{\dagger}$, respectively. $|\mathcal{A}|$ is the number of elements of set $\mathcal{A}$. Finally, the operator $\odot$  denotes the Hadamard product and $\delta[\cdot]$ represents the Dirac function. 

\section{System Model}\label{S2}

In this section, we first introduce the proposed NOMA-OTFS transmission protocol in detail. Then the training sequence aided OTFS (TS-OTFS) modulation/demodulation architecture is presented.

\subsection{General Principle of NOMA-OTFS}\label{S2.1}

As illustrated in Fig.~\ref{fig1:system model}, we consider a prospective communication scenario, where LEO satellite provides seamless coverage services to numbers of single-antenna IoT terminals.
The LEO satellite is equipped with an uniform planar array (UPA) composed of  $P = P_x \times P_y$ antennas, where $P_x$ and $P_y$ are the number of antennas on the x- and y-axes respectively. 
In light of the sporadic traffic property in typical IoT \cite{Ke.Massive access}, within a given time interval, the number of active IoT terminals $K_a$ is small compared to all potential IoT terminals $K$, i.e., $K_a \ll K$. 
The active IoT terminals share the non-orthogonal resource allocation and directly transmit their data packets arranged in DD domain lattice and pilot sequences without applying for the grant and the inactive remain sleep. 
To represent the activity state of these terminals, we define activity indicator $\alpha_k$, which is equal to 1 when the $k$-th IoT terminal is active and 0 otherwise, and the set of active terminals as $\mathcal{A} = \{k|\alpha_k = 1, 1 \le k \le K\}$.

Since there are few propagation scatters in the TSL, and the lines connecting the LEO satellite and different IoT terminals are hardly blocked by obstacles, it's reasonable to assume it's a line-of-sight (LoS) link.  
Additionally, since different IoT terminals are spatially far apart, the link between LEO satellite and different IoT terminals can be viewed as approximately uncorrelated.
Therefore, with reference to \cite{You.LEO}, the discrete form of uplink delay domain channel between the $p$-th receive antenna of LEO satellite and the $k$-th IoT terminals at the instant $\kappa$ can be represented by
\begin{align}\label{eq:discrete TSL}
	\begin{split}
		h_{k,p}[\kappa,\ell] 
		& = g_{k} e^{j2\pi\frac{\upsilon_k(\kappa-\ell_k)}{N(M+M_t)}} \delta[\ell-\ell_k] \cdot [\mathbf{v}_k]_p	
	\end{split},
\end{align}
where $g_k$, $\upsilon_k$, $\ell_k$ and $\mathbf{v}_k  \in \mathbb{C}^{P \times 1}$ denotes the small-scale fading factor, the Doppler shift, the propagation delay, and the array steering vector respectively.

\begin{figure}[tp]	
	\centering
	\includegraphics[width=\columnwidth, keepaspectratio]{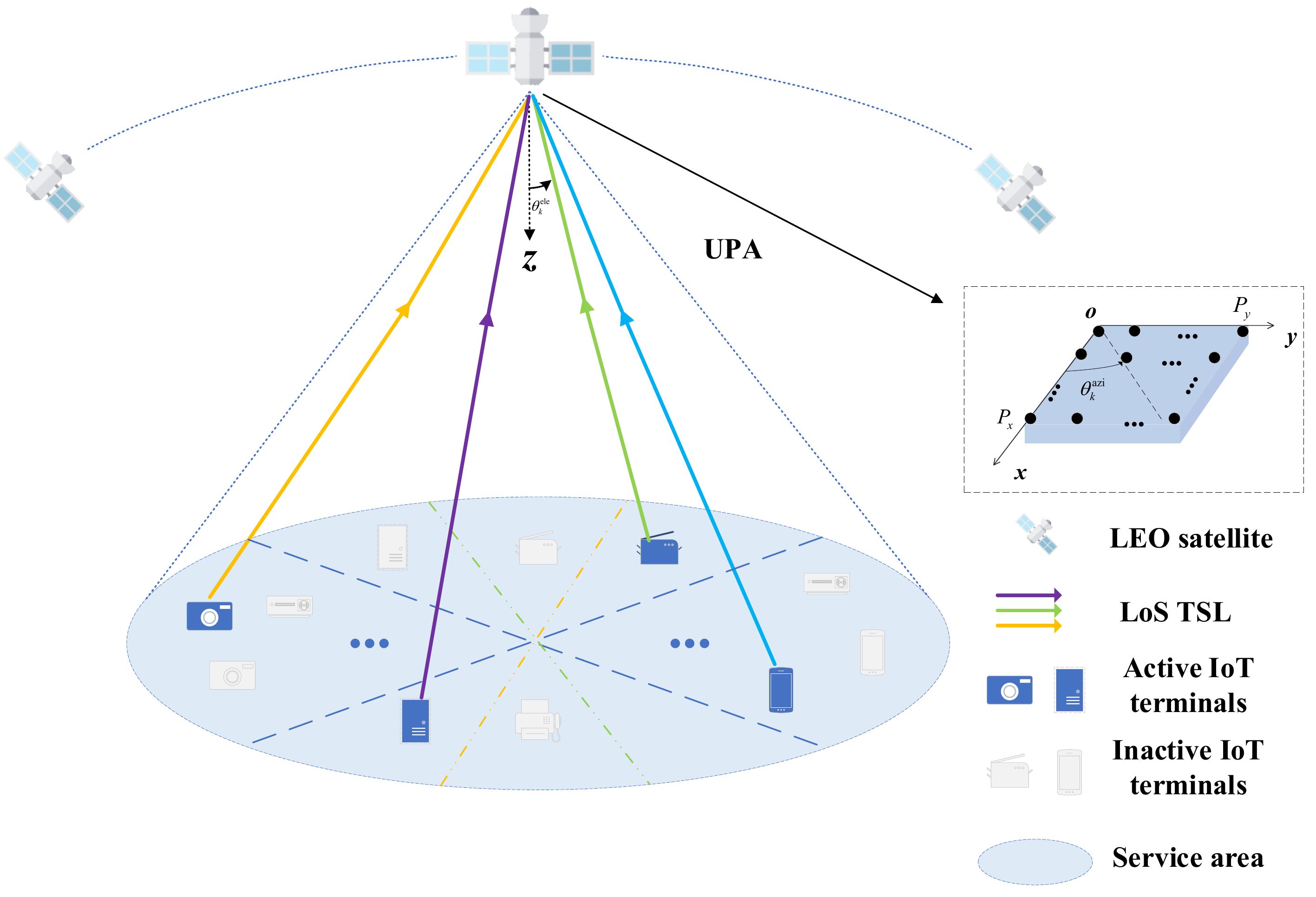}
	\caption{Illustration of the LEO constellation-enabled IoT. }
	\label{fig1:system model}	
\end{figure}

\subsection{TS-OTFS Modulation and Demodulation Architecture}\label{S2.2}
 
Next, we briefly discuss the TS-OTFS modulation and demodulation architecture as follows.

\subsubsection{TS-OTFS Modulation}\label{S2.2.1}

The input information bits are first modulated by quadrature amplitude modulation (QAM) and then rearranged into OTFS data symbols $\mathbf{X}^{\rm DD}_k \in \mathbb{C}^{M \times N}$ in the DD plane, where $N$ and  $M$ are the number of time slots and sub-carriers respectively. 
Then, the OTFS symbols are mapped into time-frequency (TF) domain symbols $\mathbf{X}^{\rm TF}_k \in \mathbb{C}^{M \times N}$ through a pre-processing module, which consists of \textit{inverse symplectic finite Fourier transform} (ISFFT) and a transmit windowing function \cite{OTFS}.  

Moreover, the TF domain symbols $\mathbf{X}^{\rm TF}_k$ are transformed to the time domain transmitted signal $\mathbf{s}_k$ by a TF modulator. Specifically, at first, \textit{Heisenberg transform} \cite{OTFS} is applied on each column of $\mathbf{X}^{\rm TF}_k$ to produce time domain symbols $\tilde{\mathbf{S}}_k \in \mathbb{C}^{M \times N}$
Next, different from the traditional OTFS system \cite{SWQ.OTFS}, where cyclic prefix is appended for each data symbol $\tilde{\mathbf{s}}_{k,i} \in \mathbb{C}^{M \times 1}$ ($\tilde{\mathbf{s}}_{k,i}$ is the $i$-th column vector of $\tilde{\mathbf{S}}_k$), we add a $M_{t}$-length known training sequence (TS) $\mathbf{c}_k = \left[c_{k,0} ~ c_{k,1} \dots ~ c_{k,M_{t}-1}\right]^{\rm T}$  at the end of data symbol, i.e. $\mathbf{s}_{k,i} = \left[ \tilde{\mathbf{s}}_{k,i}^{\rm T} \; \mathbf{c}^{\rm T}_k \right ]^{\rm T}$.  

Finally, the transmit signal $\mathbf{s}_k \in \mathbb{C}^{(M_t+M)N \times 1}$ is obtained through parallel to serial conversion.

\subsubsection{TS-OTFS Demodulation}\label{S2.1.3}

The $\kappa$-th element of the received signal $\mathbf{r}_p$ from the $p$-th receive antenna, which is the superposition of signals from different active terminals, can be expressed as  
\begin{align}\label{eq:IO} 
	r_p(\kappa) = \sum_{k=1}^K \sum_{l=0}^{L}  \alpha_k h_{k,p}\left[\kappa,\ell\right] s_k\left[\kappa-\ell\right] + w_p(\kappa),	 
\end{align}
where $w_p(\kappa) \sim \mathcal{CN}(0,\sigma^2_w)$ denotes the additive white Gaussian noise (AWGN) at the receiver of LEO satellite.

In the LEO satellite end, the demodulation process is equivalent to the inverse operation of the modulation part. Due to the limited space, its detailed explanation can be referred to \cite{SWQ.OTFS}.

\section{Proposed Joint AUD and CE Scheme}

\begin{figure*}[t]	
	\centering
	\includegraphics[width=2\columnwidth, keepaspectratio]{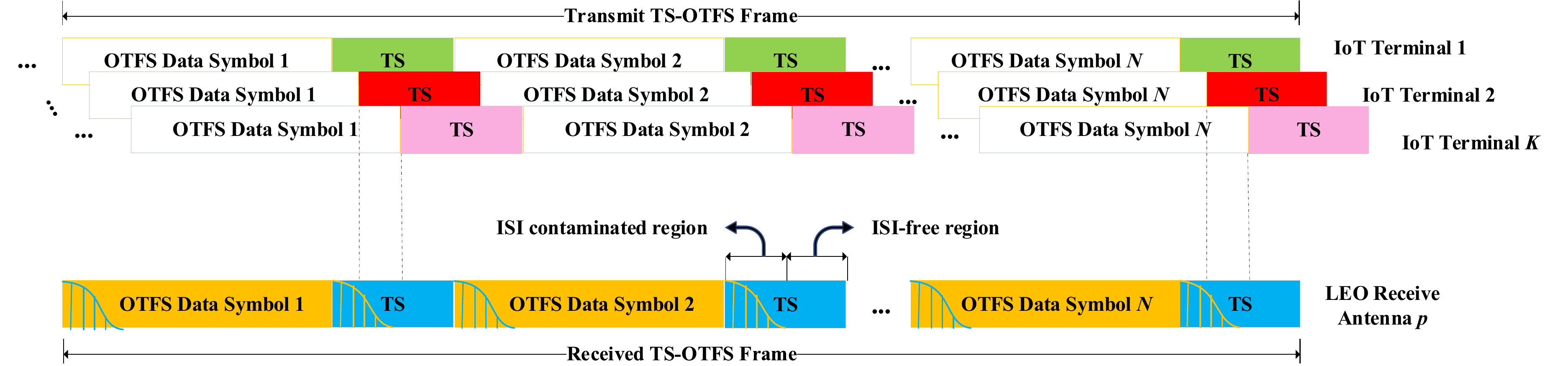}
	\caption{Structure of the transmit and the received TS-OTFS data frame.}
	\label{fig2:data frame}	
\end{figure*}

In fact, the embeded time domain TS's are alternatives of pilot and guard symbols used to perform AUD and CE. 
With the aid of TS's, the AUD and coarse CE can be formulated as a structured sparse signal recovery problem by exploiting the sporadic traffic and inherent sparsity of TSL channel described above.
In this section, we first formulate this problem and solve it under the framework of CS. 
Besides, in order to facilitate more accurate CE, we develop a parametric approach and seek to estimate the parameters fully characterizing TSL channel based on the recovered sparse sampled channel vectors.
With the estimate of these parameters, the channel impulse response (CIR) modeled in (\ref{eq:discrete TSL}) can be reconstructed and more accurate results of CE can be given.  

\subsection{Formulation of AUD and coarse CE}\label{S3.1}

The structure of $\mathbf{s}_k$ and $\mathbf{r}_p$ are illustrated in Fig. \ref{fig2:data frame}, both of which consist of OTFS data symbols and inserted time domain TS's.
When TS's are used to perform AUD and coarse CE, there is a significant challenge that the received TS's are contaminated by the previous OTFS data symbol, whose influence is significant due to long relative propagation delay especially in LEO communication system. 
An effective approach is to utilize the inter-symbol-interference free (ISI-free) region, which is the rear part of the TS's and immune from the influence of the previous data symbol, to perform AUD and coarse CE \cite{IBI-free}.
Therefore, the $M_t$-length TS is designed to be longer than the maximum  relative propagation delay $L$ in order to ensure the reliable system performance and the length of ISI-free region can be denoted as $G \triangleq M - L + 1$.  

In this way, the ISI-free region within the $i$-th received TS $\mathbf{r}_{{\rm TS},p}^{i} \in \mathbb{C}^{G \times 1}$ can be represented based on (\ref{eq:IO}) by 
\begin{align}\label{eq:vec IO}    
	\mathbf{r}_{{\rm TS},p}^{i} = \sum_{k=1}^K \alpha_k \mathbf{\Delta}_k \mathbf{\Psi}_{k} \mathbf{h}_{{\rm TS},k,p}^{i} + \mathbf{w}_{{\rm TS},p}^{i}, 
\end{align}
where $\mathbf{h}_{{\rm TS},k,p}^{i} \in \mathbb{C}^{L \times 1}$ is the vector form of CIR with a unique non-zero element $h_{k,p}[L-M_t+i(M+M_t),\ell_k]$ at the position of $\ell_k+1$, $\mathbf{w}_{{\rm TS},p}^{i} \in \mathbb{C}^{G \times 1}$ is the corresponding noise vector, 
$\mathbf{\Delta}_k = \mathrm{diag}\{ e^{j2\pi\frac{\upsilon_k \cdot 0}{N(M+M_t)}}, e^{j2\pi\frac{\upsilon_k \cdot 1}{N(M+M_t)}}, \dots, e^{j2\pi\frac{\upsilon_k \cdot (G-1)}{N(M+M_t)}} \}$ is the diagonal matrix of Doppler shift,
and $\mathbf{\Psi}_k \in \mathbb{C}^{G \times L}$ is a Toeplitz matrix given by \cite{IBI-free} 
\begin{align}\label{eq:Toeplitz} 
	\mathbf{\Psi_k} =& \left[ \begin{array}{cccc}
		c_{k,L-1} & c_{k,L-2} & \cdots & c_{k,0} \\
		c_{k,L} & c_{k,L-1} & \cdots & c_{k,1} \\
		\vdots & \vdots & \vdots & \vdots \\
		c_{k,M_t-1} & c_{k,M_t-2} & \cdots & c_{k,M_t-L} \\
	\end{array} \right].
\end{align}	

Since $\mathbf{\Delta}_k$ is an unknown matirx for the receiver of LEO
satellite, sparse CIR vectors in (\ref{eq:vec IO}) cannot be recovered with an unknown sensing matrix.
Fortunately, the duration of each ISI-free region is always over a relatively short span of time compared to the whole data frame and thus the TSL can be assumed as invariant without introducing large errors.
In this case, $\mathbf{\Delta}_k$ is approximately an identity matrix and (\ref{eq:vec IO}) can be rewritten as 
\begin{align}\label{eq:approx}    
	\mathbf{r}_{{\rm TS},p}^{i} \approx \sum_{k=1}^K \alpha_k \mathbf{\Psi}_{k} \mathbf{h}_{{\rm TS},k,p}^{i} + \mathbf{w}_{{\rm TS},p}^{i} = \mathbf{\Psi} \tilde{\mathbf{h}}_{{\rm TS},p}^i + \mathbf{w}_{{\rm TS},p}^{i},
\end{align}
where 
$\tilde{\mathbf{h}}^i_{\rm {TS},p} = \left[ \tilde{\mathbf{h}}_{{\rm TS},1,p}^{i,\rm T}, \tilde{\mathbf{h}}_{{\rm TS},2,p}^{i,\rm T}, \dots, \tilde{\mathbf{h}}_{{\rm TS},K,p}^{i,\rm T} \right]^{\rm T} \in \mathbb{C}^{KL \times 1}$ $(\tilde{\mathbf{h}}_{{\rm TS},k,p}^{i} = \alpha_k \mathbf{h}_{{\rm TS},k,p}^{i})$ and 
$\mathbf{\Psi} = \left[\mathbf{\Psi}_1,\mathbf{\Psi}_2,\dots,\mathbf{\Psi}_K\right]  \in \mathbb{C}^{G \times KL}$ 
are composed of sensing matrices and CIR vectors of different terminals respectively. 

Although both the amplitude and the phase of channel change over time and TSL channel related to different terminals vary significantly, there exists spatial correlation at the LEO satellite receiver, which means for different receive antennas, the propagation delay and the Doppler shift of the transmitted signal from the same terminal are roughly the same. 
As a result, the support sets are identical for CIR vectors $\mathbf{h}_{{\rm TS},k,p}^{i}$ with different subscripts $p$, whereas the non-zero coefficients can be different. 
Additionally, since in the duration of one frame, the relative positions of IoT terminals and LEO satellite don't change significantly, it's reasonable to assume the propagation delay and the Doppler shift of the transmit signal from the same terminal are also roughly the same  for different time slots in the one frame.   
Hence, the support sets are identical for CIR vectors $\mathbf{h}_{{\rm TS},k,p}^{i}$ with different superscripts $i$.

With above discussion in mind, we intend to extend (\ref{eq:approx}) to the form of multi-time slots and multi-antennas. 
The multi-antennas form of (\ref{eq:approx}) can be straightly rewritten as 
\begin{align}\label{eq:approx2}  
	\mathbf{R}_{\rm TS}^i \approx \mathbf{\Psi}\tilde{\mathbf{H}}^i_{\rm TS} + \mathbf{W}_{\rm TS}^i,	
\end{align}
where $\mathbf{R}_{\rm TS}^i$, $\tilde{\mathbf{H}}^i_{\rm TS}$ and $\mathbf{W}_{\rm TS}^i$ are made up of corresponding column vectors related to different antennas, i.e., $\mathbf{R}_{\rm TS}^i = \left[ \mathbf{r}_{{\rm TS},1}^{i}, \mathbf{r}_{{\rm TS},2}^{i} ,\dots, \mathbf{r}_{{\rm TS},P}^{i}  \right]  \in \mathbb{C}^{G \times P}$, $\tilde{\mathbf{H}}^i_{\rm TS} = \left[ \tilde{\mathbf{h}}_{{\rm TS},1}^i, \tilde{\mathbf{h}}_{{\rm TS},2}^i, \dots, \tilde{\mathbf{h}}_{{\rm TS},P}^i \right]  \in \mathbb{C}^{KL \times P}$, and $\mathbf{W}_{\rm TS}^i = \left[   \mathbf{w}_{{\rm TS},1}^{i}, \mathbf{w}_{{\rm TS},2}^{i}, \dots, \mathbf{w}_{{\rm TS},P}^{i} \right]  \in \mathbb{C}^{G \times P}$. 
Furthermore, by stacking received signals with different superscript $i$, the multi-time slots form of (\ref{eq:approx2}) can be written as 
\begin{align}\label{eq:approx3}  
	\mathbf{R}_{\rm TS} \approx \mathbf{\Psi}\tilde{\mathbf{H}}_{\rm TS} + \mathbf{W}_{\rm TS},	
\end{align}
where  $\mathbf{R}_{\rm TS}$, $\tilde{\mathbf{H}}_{\rm TS}$ and $\mathbf{W}_{\rm TS}$ are made up of corresponding matrices associated with different time slots, i.e. $\mathbf{R}_{\rm TS} = \left[ \mathbf{R}_{\rm TS}^{(1)},\mathbf{R}_{\rm TS}^{(2)},\dots,\mathbf{R}_{\rm TS}^{(N)} \right]  \in \mathbb{C}^{G \times NP}$, $\tilde{\mathbf{H}}_{\rm TS} = \left[ \tilde{\mathbf{H}}^{(1)}_{\rm TS},\tilde{\mathbf{H}}^{(2)}_{\rm TS},\dots,\tilde{\mathbf{H}}^{(N)}_{\rm TS} \right]  \in \mathbb{C}^{KL \times NP}$, and $\mathbf{W}_{\rm TS} = \left[ \mathbf{W}_{\rm TS}^{(1)},\mathbf{W}_{\rm TS}^{(2)}, \dots, \mathbf{W}_{\rm TS}^{(N)}\right]  \in \mathbb{C}^{G \times NP}$.

\subsection{AUD and coarse CE}\label{S3.2} 

For the joint sparse signal recovery of (\ref{eq:approx3}), many signal recovery algorithms have been proposed under the distributed compressive sensing (DCS) framework, whose goal is to jointly recover a set of vectors that share a common support rather than trying to recover a single sparse vector, and thus the accuracy of signal recovery is highly improved. 
Here, we utilize the well-known greedy algorithm simultaneous orthogonal matching pursuit (SOMP) \cite{SOMP} to recover the sparse CIR vectors of  $\tilde{\mathbf{H}}_{\rm TS}$ and its detailed instruction can be found in \cite{SOMP}. We use $\mathcal{I}$ and $\tilde{\mathbf{H}}_{\rm TS}^{e}$ to denote the recovered index of support set and the estimate of $\tilde{\mathbf{H}}_{\rm TS}$ respectively.

In accordance with the index of support set $\mathcal{I}$, the the index of support set associated with the $k$-th IoT terminal can be denoted as
${\Omega}_{k} = \{\omega_{k} |\omega_{k}\in \mathcal{I}, (k-1)L \le \omega_{k} <  kL \} $ and thus the estimate of terminals activity indicator $\hat{\alpha}_k$ can be acquired from
\begin{align}\label{eq:AUD}
	\hat{\alpha}_k = \left\{ \begin{array}{ll}
		1 & \textrm{if ${\Omega}_{k} \neq \emptyset$}\\
		0 & \textrm{otherwise}\\
	\end{array} \right..
\end{align}
As a result, the set of terminals that are identified as active can be represented by $\hat{\mathcal{A}} = \{k|\hat\alpha_k = 1, 1 \le k \le K\}$ and the number of elements in te set is denoted by $\hat{K}_a = |\hat{\mathcal{A}}|$. 
Besides, the information of their corresponding propagation delay is contained in the index of support set $\mathcal{I}$ as well.  
It's noteworthy that although the real propagation delay is single-valued for  each terminal-receive antenna pair under the LoS satellite communication environment, the estimated propagation delay can be muti-valued due to the wrong selection of support set with SOMP. As a result, the propagation delay related to terminal $k$ ($k \in \hat{\mathcal{A}}$) can be denoted as
\begin{align}\label{eq:delay}
	\hat{\ell}_{k}^q = \omega_{k}^q-(k-1)L,	
\end{align}
where $\omega_{k}^q \; (1 \le q \le |{\Omega}_{k}|)$ is the $q$-th element of ${\Omega}_{k}$.

So far, we have recovered the sparse sampled CIR vectors in $\tilde{\mathbf{H}}_{\rm TS}$ in preparation for the subsequent Doppler shift estimation.
At the same time, propagation delay $\ell_k$ and activity indicator $\alpha_k$ are obtained as  $\hat{\ell}_{k}^q$ and $\hat{\alpha}_k$ respectively for the following more accurate channel estimation.

\subsection{More accurate CE with parametric approach}\label{S3.3}

It can be observed that the CIR modeled in (\ref{eq:discrete TSL}) is complex exponential signal varying with discrete index $\kappa$. 
In fact, once their sampled values at different instants are obtained, the Doppler shift can be estimated according to Nyquist criterion. 
Fortunately, these sampled values are involved in the non-zero elements of recovered CIR vectors in $\tilde{\mathbf{H}}_{\rm TS}^{e}$. 
For the convenience of following Doppler shift estimation, we define an effective CIR matrix $\hat\Upsilon_{k}^q \in \mathbb{C}^{N \times P}$, which is reshaped from the corresponding support set of $\tilde{\mathbf{H}}_{\rm TS}^{e}$, i.e.,
\begin{align}\label{eq:eff CIR}
	\hat\Upsilon_{k}^q = \mathrm{unvec} \left \{ \left[ \tilde{\mathbf{H}}_{\rm TS}^{e} \right]_{\omega_{k}^q,:}^{\rm T} \right \}.
\end{align}

\SetAlgoNoLine
\SetAlFnt{\small}
\SetAlCapFnt{\normalsize}
\SetAlCapNameFnt{\normalsize}
\begin{algorithm}[tp!]
	\caption{TLS-ESPRIT Based Doppler Shift Estimation.} 
	\label{Alg:1}
	\KwIn{
	1) Recovered effective CIR matrix $\hat\Upsilon_{k}^q$;\\
	~~~~~~~~~~2) Number of time slots $N$;\\
	~~~~~~~~~~3) Number of measurements $P$.
}
	\KwOut{Estimate of Doppler shift $\hat\upsilon_{k}^q(\forall k \in \hat{\mathcal{A}},q \in |{\Omega}_{k}| )$.}
	
	{\bf 1.~Initialization}:\\
	Divide two subarrays for each measurement as $\mathbf{x}_{1,p} = \left[ \hat\Upsilon_{k}^q \right]_{1:N-1,p}$, $\mathbf{x}_{2,p} = \left[ \hat\Upsilon_{k}^q \right]_{2:N,p}$ 
	\\ and their combinations $\mathbf{x}_p = \left[ \mathbf{x}_{1,p}^{\rm T} \; \mathbf{x}_{2,p}^{\rm T} \right]^{\rm T} \in  \mathbb{C}^{2(N-1) \times 1}$. \\  

	{\bf 2. Denoising operation}:\\
	Take the minimum eigenvalue $\sigma^2$ of the noisy covariance matrix ${\mathbf{R}}_{xx} = E\left[\mathbf{x}\mathbf{x}^{\rm H} \right] \approx \frac{1}{P}\sum_{p=1}^{P}  \mathbf{x}_{1,p} \mathbf{x}_{2,p}$ as the estimate of noise variance, and the noise cancelled covariance matrix is given by $\hat{\mathbf{R}}_{xx} = \mathbf{R}_{xx}-\sigma^2 \mathbf{I} $. \\
	
	{\bf 3. Subspace solution}:\\
	The subspace of subarray $\mathbf{x}_{1,p}$ and $\mathbf{x}_{2,p}$ are obatained as $\mathbf{e}_{x_1} = [\hat{\mathbf{U}}_s]_{1:(N-1);1}$ and $\mathbf{e}_{x_2} = [\hat{\mathbf{U}}_s]_{N:2(N-1);1}$
	by solving the eigenvector of $\hat{\mathbf{R}}_{xx} = \hat{\mathbf{U}}_s \hat{\mathbf{\Sigma}}_s \hat{\mathbf{U}}^{\rm H}_s$. \\
	
	{\bf Result}:\\
	Solve the eigenvector of $[\mathbf{e}_{x_1}\mathbf{e}_{x_2}]^{\rm H} [\mathbf{e}_{x_1}\mathbf{e}_{x_2}] = \mathbf{E} \mathbf{\Sigma} \mathbf{E}^{\rm H}$ and partition eigenvalue matrix $\mathbf{E} \in \mathbb{C}^{2 \times 2}$ into four submatrices as $\mathbf{E} \triangleq \begin{bmatrix}e_{11} & e_{12}\\e_{21} &e_{22}\end{bmatrix}$. The estimated Doppler shift $\hat\upsilon_{k}^q$ can be calculated by $\hat\upsilon_{k}^q = \frac{N}{2\pi}\arg(-e_{12}e_{22}^{-1})$.
	
\end{algorithm}

The Doppler shift estimation algorithm based on total least squares criterion of estimating signal parameters via rotational invariance techniques (TLS-ESPRIT) \cite{ESPRIT} is presented in \textbf{Algorithm \ref{Alg:1}}. 
In fact, each column of effective CIR matrix $\hat\Upsilon_{k}^q$ is composed of CIR of $N$ different time slots, e.g., $\left[ \hat\Upsilon_{k}^q \right]_{:,p} = \big[ \hat{h}_{k,p}[L+M,\ell_k],\hat{h}_{k,p}[L+M_t+2M,\ell_k],\dots,\hat{h}_{k,p}[L-M_t+N(M+M_t),\ell_k] \big]^{\rm T}$. 
Moreover, different column vectors of effective CIR matrix $\hat\Upsilon_{k}^q$ come from different receive antennas, where the Doppler shift is approximately the same and thus they can be regarded as multiple measurements to mitigate the effects of noise.
Before proceeding with the estimation of the fading factor, we introduce a lemma as follows.

\begin{lemma} \label{lemma1}	
	Here, it's assumed that the support set recovers perfectly and thus the subscript $q$ can be omitted. We define the non-zero elements of channel vector as 
	$\hat{\mathbf{h}}_{p}^{{\rm eff},i} = \left[ \tilde{\mathbf{H}}_{\rm TS}^e \right]_{\mathcal{I},(i-1)P+p} \in \mathbb{C}^{K_a \times 1}$ and effective fading factor $g_{k,p}^{\rm eff} = g_{k} \cdot [\mathbf{v}_k]_p$. Their relationship can be written as      
	\begin{align}\label{eq:lemma}
		\hat{\mathbf{h}}_{k}^{{\rm eff},i} =  \boldsymbol{\Gamma}\boldsymbol{\eta^{i}} \odot \mathbf{g}^{\rm eff}_p+(\boldsymbol{\Psi}^{\rm H}_{\mathcal{I}} \boldsymbol{\Psi}_{\mathcal{I}})^{-1} \boldsymbol{\Psi}_{\mathcal{I}}^{\rm H} \mathbf{w}_{{\rm TS},p}^{i},	
	\end{align}
	where  $\boldsymbol{\Gamma} = \boldsymbol{\Psi}^{\dagger}_{\mathcal{I}} 
	\left[ \mathbf{\Delta}_{k_1} \boldsymbol{\psi}_{k_1},\mathbf{\Delta}_{k_2} \boldsymbol{\psi}_{k_2},\dots,\mathbf{\Delta}_{k_{K_a}} \boldsymbol{\psi}_{k_{K_a}} \right] $, $\mathbf{g}^{\rm eff}_p = \left[ g_{k_1,p}^{\rm eff}, g_{k_2,p}^{\rm eff},\dots,g_{k_{K_a},p}^{\rm eff} \right]^{\rm T}$ with  $k_1,k_2,\dots,k_{K_a} \in \mathcal{A}$,
    $\boldsymbol{\eta^{i}}$ is a Doppler shift vector and its specific expression can be found (\ref{eq:proof}), and $\boldsymbol{\psi}_{k_1},\boldsymbol{\psi}_{k_2},\dots,\boldsymbol{\psi}_{k_{K_a}}  \in \mathbb{C}^{G \times 1}$ are the support vectors associated with the active terminals in $\mathcal{A}$.
\end{lemma}
\begin{proof}
	Under the assumption of perfect selection of support set in sparse vector recovery,  $\hat{\mathbf{h}}^{\rm eff}_{k,p}$ can be derived from SOMP as 
	\begin{align}
		\hat{\mathbf{h}}^{{\rm eff},i}_{p} 
		& = \boldsymbol{\Psi}^{\dagger}_{\mathcal{I}}  \mathbf{r}_{{\rm TS},p}^{i}, 
	\end{align}
	Moreover, combined with the (\ref{eq:vec IO}), the relationship of $\hat{\mathbf{h}}_{p}^{{\rm eff},i}$ and $g_{k,p}^{\rm eff}$ can be rewritten as (\ref{eq:proof}).
	This completes the proof.

\end{proof}

\newcounter{mytempeqncnt}
\begin{figure*}[!t]
	\normalsize
	\setcounter{mytempeqncnt}{\value{equation}}
	\begin{align}\label{eq:proof}
	\begin{split}	
		\hat{\mathbf{h}}^{{\rm eff},i}_{p}
		& = \boldsymbol{\Psi}^{\dagger}_{\mathcal{I}}   \sum_{k \in \mathcal{A}} \mathbf{\Delta}_k \boldsymbol{\psi}_k \mathbf{h}_{{\rm T}S,k,p}^{i}(\ell_k+1) +  
		\underbrace { (\boldsymbol{\Psi}^{\rm H}_{\mathcal{I}} \boldsymbol{\Psi}_{\mathcal{I}})^{-1} \boldsymbol{\Psi}_{\mathcal{I}}^{\rm H} \mathbf{w}_{{\rm TS},p}^{i}}_{\mathrm{noise}}  \\
		& \approx  \overbrace{\boldsymbol{\Psi}^{\dagger}_{\mathcal{I}} 
			\left[ \mathbf{\Delta}_{k_1} \boldsymbol{\psi}_{k_1},\mathbf{\Delta}_{k_2} \boldsymbol{\psi}_{k_2},\dots,\mathbf{\Delta}_{k_{K_a}} \boldsymbol{\psi}_{k_{K_a}} \right] }^{\boldsymbol{\Gamma}} 
		 \left[ \mathbf{h}_{{\rm TS},k_1,p}^{i}(\ell_{k_1}+1),\mathbf{h}_{{\rm TS},{k_2},p}^{i}(\ell_{k_2}+1),\dots,\mathbf{h}_{{\rm TS},k_{K_a},p}^{i}(\ell_{k_{K_a}}+1) \right]^{\rm T} \\
		& = \boldsymbol{\Gamma}  \underbrace{ \left[ e^{j2\pi\upsilon_{k_1}\frac{L-M_t+i(M+M_t)-\ell_{k_1}}{N(M+M_t)}},e^{j2\pi\upsilon_{k_2}\frac{L-M_t+i(M+M_t)-\ell_{k_2}}{N(M+M_t)}},\dots,e^{j2\pi\upsilon_{k_{K_a}}\frac{L-M_t+i(M+M_t)-\ell_{k_{K_a}}}{N(M+M_t)}} \right]^{\rm T} }_{\boldsymbol{\eta^{i}}} \\
		& \quad
		 \odot \left[ g_{k_1,p}^{\rm eff}, g_{k_2,p}^{\rm eff},\dots,g_{k_{K_a},p}^{\rm eff} \right]^{\rm T} \\
	\end{split}
 \end{align}
	\setcounter{equation}{\value{mytempeqncnt}}
	\stepcounter{equation}
	\hrulefill
	\vspace*{4pt}
\end{figure*}

In line with the \textbf{Lemma \ref{lemma1}}, with the recovered Doppler shift $\hat\upsilon_{k}^q$ and propogation delay $\hat\ell_{k}^q$, the effective channel fading factor can be calculated.
So far, necessary parameters to reconstruct the CIR have been acquired and the more accurate result of CE can be given by
\begin{align}\label{eq:reconstruct}   
	\hat{h}_{k,p}[\kappa,\ell] = \sum_{q=1}^{|\boldsymbol{\Omega}_{k}|} \hat{g}_{k,p}^{{\rm eff},q} e^{j2\pi\frac{\hat\upsilon_{k}^q(\kappa-\hat{\ell}_{k}^q)}{N(M+M_t)}} \cdot \delta[\ell-\hat{\ell}_{k}^q].
\end{align}

\section{Simulation Results}

In this section, we provide simulation results to prove the effectiveness of our proposed scheme. First of all, to evaluate the performance fairly, we define the activity error rate (AER) of AUD 
\begin{align}\label{eq:AER}
	\mathrm{AER} = \frac{1}{K} \sum_{k=1}^{K} |\hat{\alpha}_k-\alpha_k|, 	
\end{align}
and the normalized mean square error (NMSE) of CE 
\begin{align}\label{eq:NMSE} 
	\mathrm{NMSE} = \frac{\sum\limits_{k,p,\kappa,\ell}   ||
		\hat{\alpha}_k \hat{h}_{k,p}[\kappa,\ell] - \alpha_k h_{k,p}[\kappa,\ell] ||^2_2}
	{\sum\limits_{k,p,\kappa,\ell} ||\alpha_k h_{k,p}[\kappa,\ell] ||^2_2}.	
\end{align}	
The detailed system parameters for simulation are summarized in Table \ref{tab1}.  

\begin{table}[h]
	\centering
	\caption{Simulation Parameters}
	\label{tab1}
		\begin{tabular}{|c|c|c|}
			\hline
			\textbf{Contents}       & \textbf{Parameters}                           & \textbf{Values}                \\ \hline
			\multirow{5}{*}{System} & Carrier frequency (GHz)                       & 10                             \\ \cline{2-3} 
			& Subcarrier spacing (kHz)                      & 480                            \\ \cline{2-3} 
			& Bandwidth (MHz)                               & 122.88                         \\ \cline{2-3} 
			& Size of OTFS frame ($M$,$N$)                  & (256,8)                        \\ \cline{2-3} 
			& Number of BS antennas ($P_x$,$P_y$)           & (10,10)                        \\ \hline
			\multirow{4}{*}{TSL channel}    & LEO satellite velocity (km/s)                 & 7.58                           \\ \cline{2-3} 
			& IoT terminal velocity (m/s)                             & $0\sim100$                     \\ \cline{2-3} 
			& Propagation delay (ms)       & $0\sim0.067$                   \\ \cline{2-3} 
			& Doppler shift (kHz)                    & $0\sim178.2$                    \\ \hline 
		\end{tabular}%
\end{table}

For comparison, we take \cite{SWQ.OTFS} as the benchmark, which utilizes the embedded DD domain pilots and guard symbols to estimate the channel.
The sizes of pilots along the Doppler dimension and the delay dimension are denoted as $N_{\nu}$ and $M_{\tau}$ respectively, and we choose $N=N_{\nu}$ as \cite{SWQ.OTFS} did. Moreover, we consider the oracle-LS with the activity states of all the terminals and the support set of CIR vectors known as the lower bound of the proposed method.

\begin{figure}[!htp]
	\centering
	\subfigure[SNR versus AER]{	
		\includegraphics[width=\columnwidth, keepaspectratio]{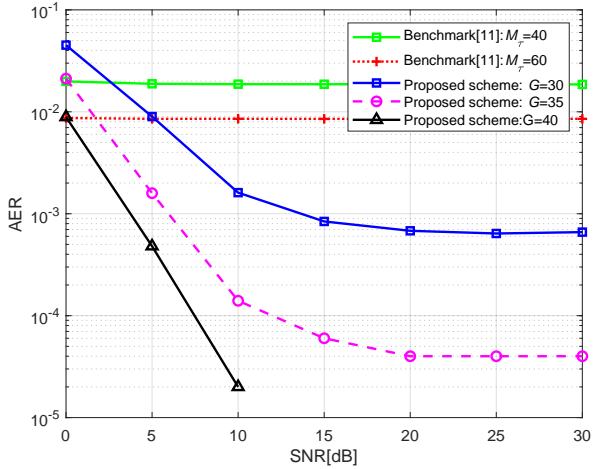}}
	
	\subfigure[SNR versus NMSE]{	
	\includegraphics[width=\columnwidth, keepaspectratio]{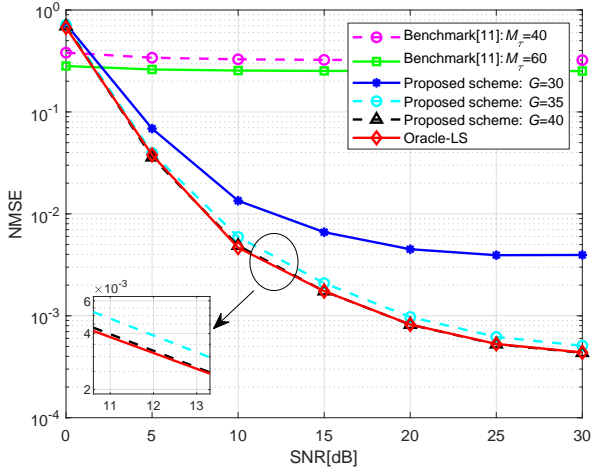}}
	
	\caption{Performances comparison against the SNR, where the number of active IoT terminals and potential terminals are fixed as $K_a = 10$ and $K=100$ respectively.}
	
	\label{fig3:comparision}
\end{figure}

As Fig. \ref{fig3:comparision} shows, with regard to simulation results of AUD and CE, on the one hand, the superiority of our method over the existing method is particularly noticeable. The scheme which performs AUD and CE in DD domain completely fails in satellite-enabled communication environment, which also confirms that it suffers from serious performance degradation under severe Doppler shift environment. 
This is because, the required pilots overhead are seriously increased to ensure a constant NMSE when Doppler shift is large and since small size of Doppler dimension resources units $N$ is usually adopted in this kind of fast time-varying channel, the resolution in Doppler domain is seriously limited and the support set is hardly right selected, which has great impact on the performance of both AUD and CE. 
On the other hand, the NMSE performance of our proposed method is very close to Oracle-LS when $G \ge 30$ in this simulation parameters configuration. 
The satisfactory CE accuracy and robustness of our proposed method under sever Doppler shift environment are owing to the accurate selection of support set with multiple measurements, super-resolution estimate of Doppler shift and the subsequent more accurate channel reconstruction.

\section{Conclusions}

In this paper, we introduce the NOMA-OTFS system to LEO-enabled constellation IoT
for facilitating reliable service with the massive connectivity and  investigate a new joint AUD and CE method for the first time. The designed TS-OTFS data frame enables us to fully leverage the sporadic traffic property and the structured sparsity of the TSL channel, as a result, we can perform AUD and coarse CE with the low cost of TS's. 
Furthermore, the super-resolution of Doppler shift in time domain and the subsequent accurate channel estimation with the parametric approach results in a significant performance improvement. 
As the simulation results shows, the superiority of our proposed method is noticeable in these high speed use cases including satellite-enabled commnuication scenario.

\section*{Acknowledgment}

The work was supported by NSFC under Grants 62071044 and 61827901, the BJNSF under Grant L182024.

\end{document}